\documentclass{llncs}

\usepackage[lncs,amsthm,autoqed]{pamas}
\usepackage{arydshln}

\newcommand\eat[1]{}

\title{Path coalitional games}

\author{%
	Haris Aziz\inst{1} \and
	Troels Bjerre S{\o}rensen\inst{2}}

\institute{%
	Institut f\"ur Informatik,
	Technische Universit\"at M\"unchen, 
	80538 M\"unchen, Germany \\
	\email{aziz@in.tum.de}
	\and 
	Computer Science Dept., University of Warwick, 
	Coventry, UK, CV4 7AL \\
	\email{trold@dcs.warwick.ac.uk}}

\begin{document}

\maketitle

\begin{abstract}
		We present a general framework to model strategic aspects and stable and fair resource allocations in networks via variants and generalizations of \emph{path coalitional games}. In these games, a coalition of edges or vertices is successful if it can enable an $s$-$t$ path. 
		We present polynomial-time algorithms to compute and verify least core payoffs of \emph{cost-based generalizations} of path coalitional games and their duals, thereby settling a number of open problems. 
		The least core payoffs of path coalitional games are completely characterized and a polynomial-time algorithm for computing the nucleolus of edge path coalitional games on undirected series-parallel graphs is presented.
		
		
\end{abstract}

\section{Introduction}

We consider simple coalitional games called \emph{path coalitional games}, in particular \emph{Edge Path Coalitional Games (EPCGs)} and \emph{Vertex Path Coalitional Games (VPCGs)}. 
In these games, the players control the edges and the vertices, respectively, and a coalition of players wins if it enables a path from the source $s$ to the sink $t$ and loses otherwise. Both of these coalitional games are natural representations, for which solution concepts such as the Shapley value or the nucleolus represent the amount of payoff the respective edges or vertices deserve for enabling a path from $s$ to $t$. The payoff can indicate the importance of the players or the proportional resource, profit, maintenance or security allocation required at the respective nodes and vertices. This kind of stability analysis is especially crucial if the underlying graph represents a logistics, communication, military, supply-chain or information network~\citep{BP10a,Nebel10a}. 
We study the computational complexity of computing important cooperative game theoretic solutions of path coalitional games.

Path coalitional games also have a natural correspondence with two-person zero-sum noncooperative games. In such games, which we term as \emph{path intercept games}, there are two players, the \emph{interceptor} and the \emph{passer}. The problem is to maximize the probability of intercepting a strategically chosen path in an undirected graph. We refer to the path intercept games as \emph{Edge Path Noncooperative Games (EPNGs)} and \emph{Vertex Path Noncooperative Games (VPNGs)}. The pure strategies of the interceptor are the edges $E$ (or vertices $V$) and the set of pure strategies of the \emph{passer} is the set $\mathcal{P}$ which contains all paths from vertex $s$ to vertex $t$. If the edge (or vertex) used by the interceptor intersects with the chosen path, then the interceptor wins and gets payoff $1$. Otherwise, the interceptor loses and gets payoff $0$. Thus, the {\em value} of the game is the greatest probability that the interceptor can guarantee for successfully intercepting the chosen path.

The area of algorithmic cooperative game theory is beset with negative computational results~\citep[see e.g., ]{DeFa08a,EGGW09a,KBBEGRZ08a}. 
In this paper, we present positive algorithmic results for \emph{cost-based generalizations} of path coalitional games and their duals. The cost-based generalization of a simple game is a rich and widely-applicable model. For example in the case of edge path coalitional games, each edge charges a certain cost for its services being utilized. A coalition of edges gets a fixed reward for enabling an $s$-$t$ path. It is then natural to examine payoffs which are fair and stable and also manage to transport goods from $s$ to $t$~\citep{FGM00a}. The cost-based generalizations of path coalitional games have significance in logistics, planning and operations research. Similarly, the cost-based generalizations of duals of path coalitional games have natural importance in proposing stable reward schemes to protect strategic assets or blocking intruders in a network.

\paragraph{Contribution:}

\begin{itemize}
	\item We use dualization and cost-based generalization to provide a unifying way to model $s$-$t$ connectivity. In doing so, we also identify some interesting connections between coalitional game theory and network interdiction. 
	\item For the cost-based generalization of path coalitional games and their duals we present the first polynomial-time algorithms to compute and verify least core payoffs. Interestingly, the problem of computing the least core of the dual of vertex path coalitional game was (wrongly) claimed to be NP-hard~\citep{BP10a} and the problem of computing the least core of cost-based generalization of edge path coalitional game was conjectured to be NP-hard
~\citep[page 65, ]{Nebel10a}. 
\item We present an  algorithmic technique to compute least core payoffs for cost-based generalizations of simple games in any representation. As a corollary, it is shown that there exist polynomial-time algorithms to compute the least core payoffs of cost-based generalizations of \emph{spanning connectivity games} and \emph{weighted voting games} with bounded weights and costs.
	\item The least core payoffs of simple path coalitional games are characterized and purely combinatorial polynomial-time algorithms to compute a least core payoff are presented.
	\item The nucleolus is a solution concept which is notoriously hard to compute for most interesting coalitional games. A polynomial-time algorithm to compute the nucleolus of edge path coalitional games for undirected series-parallel graphs is presented. 
	
	
\end{itemize}

\section{Related work}

Network interdiction is the general framework in which weakening of a network by an adversary or fortification of a network by defenders is considered~\cite{SmithLim08a}. Within this body of literature, shortest path interdiction (in which an adversary wants to maximize the length of the $s$-$t$ shortest path in a directed network) is related to our setting of $s$-$t$ path coalitional games. Whereas all the variants of shortest path interdiction problem are NP-hard~\citep{KBBEGRZ08a}, we present positive computational results. 

While \emph{Vertex Path Noncooperative Games (VPNGs)} have not been considered in the interdiction literature (to the best of our knowledge), \emph{Edge Path Noncooperative Games (EPNGs)} is equivalent to the two player zero sum games considered in \citep{WashWood95a}.  \citet{WashWood95a} studied maxmin strategies in EPNGs. Our cooperative game formulation helps us in proposing equilibrium refinements such as the nucleolus (which corresponds to a unique refinement of the maxmin strategy of the corresponding path intercept games) and other cooperative-game solution concepts such as the Shapley value. 
The coalitional model, especially the cost-based generalizations, helps us reason about more elaborate security settings in which incentives, money, and cooperation of agents is involved.
The comparison between a simple coalitional game and its natural noncooperative version is similar in spirit to \citep{ALPS09a} where spanning trees are considered. 

The general definition of cost-based generalization of a simple game is inspired by \citet{FGM00a} and \citet{BP10a} where cost-based versions of specific graph-based simple games are considered. 
\citet{FGM00a} examined conditions for the non-emptiness of the core of \emph{`shortest path games'} (which are equivalent to cost-based generalization of EPCGs). However, computational problems such as computing a least core payoff were not considered.  
Different variants of EPCGs were considered under the umbrella of `shortest path games' in \citep{Nebel10a} but either the complexity of core-based relaxations is not examined or the complexity of computing least core solutions was left open and in fact conjectured to be computationally hard~\citep[page 65, ]{Nebel10a}. 
Similarly, in \citep{BP10a}, it is claimed that the least core of the dual of VPCGs is NP-hard to compute. We disprove the claim in \citep{BP10a} and present a polynomial-time algorithm to solve a generalization of the same problem for this game as well as three other games on \emph{any graph}.

A variant of EPCGs was also considered in \citep{NiRo01a} but the focus was on strategy-proof mechanisms rather than stability issues. 
\citet{DIN99a} consider a different type of $s$-$t$ connectivity game which is balanced.

The $s$-$t$ path connectivity setting also has natural links with network reliability where the goal is to compute the probability that there exists a connected path. 
However the network reliability literature does not consider strategic settings and certainly has no equivalent concepts such as the least core and the nucleolus etc.

\section{Preliminaries}

In this section, we first define the path coalitional games and path intercept games and then consider suitable game-theoretic solution concepts for these games. 

\subsection{Games}

We begin with the formal definition of a \emph{coalitional game}.

\begin{definition}[Coalitional games]	
A \emph{coalitional game} is a pair $(N,v)$ where $N=\{1,\ldots, n\}$ is a set of players and $v:2^N \rightarrow \mathbb{R}_+$ is a \emph{characteristic or valuation function} that associates with each coalition $S\subseteq N$ a payoff $v(S)$ where $v(\varnothing)=0$.%
\footnote{Throughout the paper, we assume $0\in\mathbb R_+$. }
A coalitional game~$(N,v)$ is \emph{monotonic} when it satisfies the property that $v(S)\leq v(T)$ if $S \subseteq T$.
\end{definition}
Throughout the paper, when we refer to a coalitional game, we assume such a coalitional game with transferable utility. For the sake of brevity, we will sometimes refer to the game $(N,v)$ as simply~$v$.

\begin{definition}[Simple game]	
A \emph{simple game} is a monotonic coalitional game $(N,v)$ with $v:2^N \rightarrow \{0,1\}$ such that $v(\varnothing)=0$ and $v(N)=1$. A coalition $S \subseteq N$ is \emph{winning} if $v(S)=1$ and \emph{losing} if $v(S)=0$. A \emph{minimal winning coalition} of a simple game $v$ is a winning coalition in which defection of any player makes the coalition losing.
\end{definition}

We now define the following two \emph{path coalitional games}.

\begin{definition}[Path coalitional games]	
	For an unweighted directed/undirected graph, $G=(V\cup\{s,t\},E)$,
	\begin{itemize}
		\item the corresponding \emph{Edge Path Coalitional Game (EPCG)} is a simple coalitional game $(N,v)$ such that $N=E$ and for a $S\subseteq N$, $v(S)=1$ if and only if $S$ admits an $s$-$t$ path.
		\item the corresponding \emph{Vertex Path Coalitional Game (VPCG)} is a simple coalitional game $(N,v)$ such that $N=V$ and for a $S\subseteq N$, $v(S)=1$ if and only if $S$ admits an $s$-$t$ path.
	\end{itemize}	

\end{definition}

\begin{definition}[Dual of a game]
For a game $G=(N,v)$, the corresponding \emph{dual} game $G^D=(N,v^D)$ can be defined in the following way:
$v^D(S)=v(N)-v(N\setminus S)$ for all $S\subseteq N$. 
\end{definition}

For both EPCG and VPCG, the corresponding duals $EPCG^D$ and $VPCG^D$ can be defined. It will be seen that $VPCG^D$ is equivalent to a well-studied coalitional game.

For a simple game, we can define a game which is the \emph{cost-based generalization}.

\begin{definition}[Cost-based generalization]
For a given simple game $G=(N,v)$ we can define a \emph{cost-based generalization} $C$-$G=(N,v^c)$ based on \emph{cost vector} $c=(c_1,\ldots, c_{|N|})\in {\mathbb{R}_{+}}^{|N|}$ and \emph{reward} $r\in \mathbb{R}_{+}$. 
For a coalition $S\subseteq N$, the value of the $v^c(S)=r-\min_{S'\subseteq S, v(S')=1}(\sum_{i\in S'}c_i)$ if $v(S)=1$ and $v^c(S)=0$ if $v(S)=0$.
\end{definition}

The intuition of a cost-based generalization is that each player demands some cost for its services being utilized and a coalition of players $S$ get a reward $r$ only if it is winning and gets the job done. The coalition also incurs a cost of $\min_{S'\subseteq S, v(S')=1}(\sum_{i\in S'}c_i)$ when it pools resources to get the job done. Based on this formulation, we can define Edge Path Coalitional Games with costs C-EPCG and Vertex Path coalitional games with costs, C-VPCG. It is easy to see that for a game C-G, if $r=1$ and the costs are all zero, then C-G is equivalent to G.

\begin{observation}
	C-ESPG is equivalent to the \emph{value shortest path game (VSPG)} in \citep{Nebel10a}. ${VPCG}^D$ is equivalent to the \emph{simple path disruption game} in \citep{BP10a}.	
\end{observation}

We now define the following two \emph{path intercept games}.

\begin{definition}[Path intercept games]
	For an unweighted directed graph, $G=(V\cup\{s,t\},E)$, 
	the corresponding \emph{Edge Path Noncooperative Game (EPNG)} is a noncooperative game with two players, the \emph{interceptor} and the \emph{passer}. The pure strategies of the interceptor are the edges $E$ and the pure strategies of the \emph{passer} is set $\mathcal{P}$ which contains all paths from vertex $s$ to vertex $t$. If the edge used by the interceptor intersects with the chosen path, then the interceptor wins and gets payoff $1$. Otherwise it loses and gets payoff $0$. 
	
	\emph{Vertex Path Noncooperative Games (VPNGs)} have an analogous definition to EPNGs except that the pure strategies of the interceptor are the vertices $V$ and that if the vertex used by the interceptor intersects with the chosen path, then the interceptor wins and gets payoff $1$.

\end{definition}

Both EPCG and VPNG can be generalized to the case with detection probabilities where the probability that the passer moving through edge $e$ (or vertex $v$) will be detected if the interceptor inspects $e$ (or $v$) is $p_e$ (or $p_v$ respectively).

\subsection{Cooperative Solutions}

A cooperative game solution consists of a distribution of the value of the grand coalition over the players. Formally speaking, a solution associates with each cooperative game~$(N,v)$ a set of \emph{payoff vectors $(x_1,\ldots,x_n)\in\mathbb R^N$} such that $\sum_{i\in N}x_i=v(N)$, where $x_i$ denotes player~$i$'s share of~$v(N)$. Such efficient payoff vectors are also called \emph{preimputations}.
As such, solution concepts formalize the notions of fair and stable payoff vectors. In what follows, we use notation similar to that of \citet{EGGW09a}.

Given a cooperative game~$(N,v)$ and payoff vector $x=(x_1,...,x_n)$, the \emph{excess
 of a coalition $S$ with respect to $x$} is defined by \[e(x,S)=x(S)-v(S)\text,
\] 
where $x(S)=\sum_{i\in S}x_i$.
We are now in a position to define one of the most fundamental solution concepts of cooperative game theory, viz., the core.

\begin{definition}[Core]
A payoff vector $x=(x_1,\ldots,x_n)$ is in the \emph{core} of a cooperative game~$(N,v)$ if and only  for all $S\subseteq N$, $e(x,S) \geq 0$.
\end{definition}
A core payoff vector guarantees that each coalition gets at least what it could gain on its own. The core is a desirable solution concept, but, unfortunately it is empty for many games. Games which have a non-empty core are called \emph{balanced}. The possibility of the core being empty led to the development of the  \emph{$\epsilon$-core}~\citep{ShSh66a} and the \emph{least core}~\cite{MPS79a}.


\index{$\epsilon$-core}

\begin{definition}[Least core]
For $\epsilon>0$, a payoff vector $x$ is in the \emph{$\epsilon$-core} if for all $S\subseteq N$, $e(x,S)\geq-\epsilon$. The payoff vector~$x$ is in the \emph{least core} if it is in the $\epsilon$-core for the smallest~$\epsilon$ for which the $\epsilon$-core is non-empty. We will denote by $-\epsilon_1(v)$, the minimum excess of any least core payoff vector of $(N,v)$.
\end{definition}

It is easy to see from the definition of the least core, that it is the solution of the following linear program (LP):

\begin{equation}
\label{LC-LP}
\begin{array}{ll}
\min & \epsilon  \\
\text{s.t.} & x(S)\geq v(S)-\epsilon\ \ \text{for all}\
S\subseteq N,\,\\
&\epsilon\geq 0, x_i\geq 0\ \text{for all}\
i\in N ,\\
& \sum_{i=1,\ldots,n} x_i = v(N)\ .\\
\end{array}
\end{equation}

The nucleolus is a special payoff vector which is in the core if the core is non-empty and is otherwise a member of the least core. The \emph{excess vector} of a payoff vector $x$, is the vector $(e(x,S_1),...,e(x,S_{2^n}))$ where $e(x,S_1)\leq e(x,S_2)\leq \ldots \leq e(x,S_{2^n})$.

\begin{definition}[Nucleolus]
A payoff vector $x$ such that $x_i\geq v(\{i\})$ for all $i\in N$ and $x$ has lexicographically the largest excess vector is called the \emph{nucleolus}. 
\end{definition}

The nucleolus is unique and always exists as long as $v(S)=0$ for all singleton coalitions~\citep{Sc69a}.

\section{Least core of path coalitional game variants}

Before considering other computational issues, we notice that the value of a coalition in EPCGs and VPCGs can be computed in polynomial time. For a coalition $S$ in a EPCG/C-VPCG, use \emph{Depth First Search} to check whether $s$ and $t$ are connected in a graph restricted to $S$. If not, then $v(S)=0$. Otherwise, $v(S)$ is equal to $1$.

Our first observation is that in all games EPCG, VPCG, $EPCG^D$ and $VPCG^D$, the core can be empty. In fact, the following proposition characterizes when the core of these games is non-empty:

\begin{proposition}
	The core of
	\begin{itemize}
		\item EPCG is non-empty if and only if there exists an edge, the removal of which disconnects $s$ and $t$.
		\item VPCG is non-empty if and only if there exists a vertex, the removal of which disconnects $s$ and $t$.
		\item $EPCG^D$ is non-empty if and only if there exists an $(s,t)$ edge.
		\item (\citet{BP10a}) $VPCG^D$ is non-empty if and only if there exists a vertex $x$ such that $(s,x)$ and $(x,t)$ are edges in the graph.
	\end{itemize}
\end{proposition}
\begin{proof}
	All cases follow directly from the fact that in a simple monotone game, the core is non-empty if and only if there exists a vetoer, i.e., a player $i\in N$ such that $v(N\setminus \{i\})=0$~\citep[see e.g., ][]{EGGW09a}. For the dual games, note the following. 
	Let $(N,v^d)$ be the dual game and let $x$ be a player such that $v(x)=1$. We want to show that player $x$ is a vetoer in $(N,v^d)$ i.e., $v^d(N\setminus x)=0$. 
	We know that $v(N)=1$. Then, by definition of dual, $v^d(N\setminus x)=v(N)-v(x)=1-1=0$. Thus x is a vetoer in $(N,v^d)$. 
	
\end{proof}

Since the core can be empty, the least core payoff assumes more importance. 
We will first present a general positive result (Theorem~\ref{th:main}) regarding the computation of least core payoff for cost-based generalizations of simple games. 
For a simple game $G=(N,v)$ and $(x_1, \ldots, x_{N})\in {\mathbb{R}_{+}}^{|N|}$, denote by $G^x$ game $G$ in which each player $i\in N$ has weight $x_i$. Then, a \emph{minimum weight winning coalition} of $G^x$ a winning coalition $S$ such that $x(S)$ is minimal.

\begin{theorem}\label{th:main}
For a simple game $G=(N,v)$, assume that there exists an algorithm which for a given 
weight vector $(x_1, \ldots, x_{N})\in {\mathbb{R}_{+}}^{|N|}$, computes in polynomial time a minimum weight winning coalition of $G^x$. Then, a least core payoff of a cost-based generalization of $G$ can be computed and verified in polynomial time.
\end{theorem}
\begin{proof}

	We will denote the algorithm in the statement as Algorithm $A$. 
	Consider C-G$=(N,v^c)$ be the cost-based generalization of $(N,v)$ with associates cost vector $c=(c_1,\ldots, c_{|N|})\in {\mathbb{R}_{+}}^{|N|}$ and reward $r\in \mathbb{R}_{+}$.

In order to compute a least core payoff of C-G, we consider the least core LP for C-G. 
The size of the linear program~\eqref{LC-LP} is exponential in the 
	size of game C-G, with an inequality for every subset
	of players. However, this linear program can be solved using the 
	ellipsoid method and a separation oracle, 
	which verifies in polynomial time 
	whether a solution is feasible or returns a violated
	constraint~\cite{Sch04a}. We now demonstrate how algorithm $A$ can be used to construct the separation oracle for the least core LP of C-G.

A candidate solution for the least core payoff is an efficient payoff $x=(x_1,\ldots, x_{|N|})$ such that $x(N)=v^c(N)=r-\min_{S'\subseteq N, v(S')=1}(\sum_{i\in S'}c_i)=r-$ weight of the minimum weight winning coalition of $G^c$. Since $c=(c_1,\ldots, c_{|N|})\in {\mathbb{R}_{+}}^{|N|}$, Algorithm $A$ can be used to compute $\min_{S'\subseteq S, v(S')=1}(\sum_{i\in S'}c_i)$ and therefore $x(N)$. Now that $x(N)$ is known, a separation oracle for the least core LP considers different candidate solutions $x$ such that $x(N)$ is constant.

For a candidate solution $x=(x_1,\ldots, x_{|N|})$, where $x(N)=v(N)=r-$ minimum cost of a winning coalition in C-G, construct the weighted function $x'=(x_1',\ldots, x_{|N|}')$ such that $x_i'=x_i+c_i$ for all $i\in N$. Since $x_i\geq 0$ and $c_i\geq 0$ for all $i\in N$, $x_i'\geq 0$ for all $i\in N$. Therefore, we can use algorithm $A$ to compute a minimum weight winning coalition $S^*$ of $G^{x'}$. 

Now the claim is that $S^*$ is a coalition with the minimum excess of C-G with respect to payoff $x$ and that $x'(S^*)-r$ is the minimum of excess of game C-G with respect to payoff $x$. Note that $e(x,S^*)=x(S^*)-v^c(S^*)\leq x(N)-v^c(S^*)=x(N)-(r- \min_{S''\subseteq S^*, v(S'')=1}c(S''))\leq x(N)-(r- \min_{S''\subseteq N, v(S'')=1}c(S''))= x(N)-v^c(N)= 0.$ 

For the sake of contradiction, assume that there is a coalition $S'$ which has the minimum excess with respect to $x$ in game C-G such that $S'$ is not a minimum weight winning coalition of $G^{x'}$. Then, either $S'$ is not winning or is winning but not a minimum weight winning coalition. If $S'$ is losing, then $e(x,S')=x(S')-v^c(S')=x(S')\geq 0$. Since $e(x,S^*)\leq 0$, $S'$ does not have smaller excess that $S^*$ in game C-G with respect to payoff $x$.

In the second case, assume that $v^c(S')=1$ but $x'(S')>x'(S^*)$. Without loss of generality, $S'$ is a minimal winning coalition. If it were not, then we prove that there exists an $S''\subset S'$ such that $S''$ is a minimal winning coalition and 
$e(x,S'')\leq e(x,S')$. If $v(S'')=v(S')$, then we are already done as $x(S'')\leq x(S)'$. Assume that $v(S')<v(S'')$. Then, there exists a minimal winning coalition $S'''\subset S'$ such that $c(S''')<c(S'')$. But then it must be that $x(S''')+c(S''')\geq x(S'')+c(S'')$ because if it were not, then $e(x,S''')\leq e(x,S'')$. Thus, we have established that $S'$ is a minimal winning coalition without loss of generality. Since $x(S')+c(S')=x'(S')> x'(S^*)=x(S^*)+c(S^*)$, therefore $e(x,S')=x(S')-(r-c(S'))=x(S')+c(S')-r=x'(S')-r>x'(S^*)-r=e(x,S^*)$. 

We can use the known algorithm $A$ to compute the winning coalition $S^*$ with the smallest total weight $x'(S^*)$. 
If we have $x'(S^*)=x(S^*)+c(S^*)-r\geq -\epsilon$, then $x(S)-v^c(S)\geq -\epsilon$ for all $S\subseteq N$. Therefore, $x$ is feasible. Otherwise, the constraint $x(S^*)-v^c(S^*)\geq -\epsilon$ is violated. This completes our argument that a polynomial-time separation oracle for the least core LP of the C-G can be constructed.
	
 A payoff $x=(x_1,\ldots, x_{|N|})$ can be verified if it is in the $\epsilon$-core by using the separation oracle. Since the minimum excess $-\epsilon_1$ of the least core payoff can be computed, therefore the separation oracle can also be used directly to check if the given payoff is in the least core.
\end{proof}

\begin{corollary}\label{cor:scg}
A least core can be computed for cost-based generalizations of the following games: \emph{spanning connectivity games}~\citep{ALPS09a} and  \emph{weighted voting games~\citep{EGGW09a} with bounded weights and also bounded costs}.
\end{corollary}
\begin{proof}
	For a spanning connectivity game $G$, there exists an algorithm which for a given 
	weight vector $(x_1, \ldots, x_{N})\in {\mathbb{R}_{+}}^{|N|}$, computes in polynomial time a minimum weight winning coalition of $G^x$~\citep[Proposition 5, ][]{ABH10a}. 
	
	For weighted voting games with weights represented in unary, there exists an algorithm which for a given 
	weight vector $(x_1, \ldots, x_{N})\in {\mathbb{R}_{+}}^{|N|}$, computes in polynomial time a minimum weight winning coalition of $G^x$~\citep[Theorem 5, ][]{EGGW09a}. Since the algorithm works only for weighted voting games with small weights, the algorithm can be used as a separation oracle for the least core of cost-based generalization of weighted voting games only if the associated cost vector is also represented in unary.	
\end{proof}

From the proof of Theorem~\ref{th:main} and Corollary~\ref{cor:scg} it is evident that if there the separation oracle to compute the least core LP of a simple game $G$ can also be used as a separation oracle to compute the least core LP of a cost-based generalization of $G$. We say that the representation of a coalitional game $(N,v)$ is \emph{as compact} as the cost function $c=(c_1,\ldots, c_{N})\in {\mathbb{R}_{+}}^{|N|}$, if the following condition holds: if cardinal values used in the representations of $(N,v)$ are in unary, then $c$ is also represented in unary. 

\begin{observation}\label{th:main2}
	Let $G=(N,v)$ be the underlying simple game and C-G$=(N,v^c)$ be the cost-based generalization of $(N,v)$. Assume that the representation of $(N,v)$ is as compact as the cost function $c$. Then, if there exists a polynomial-time separation oracle for the least core LP of the underlying simple game, then a least core payoff of the cost-based generalization can be computed and verified in polynomial time.
\end{observation}
 
We now apply Theorem~\ref{th:main} to path coalitional games.

\begin{theorem}\label{th:escg-costs-lc}
There exist polynomial-time algorithms to compute and verify least core payoffs of cost-based generalizations of Edge Path Coalitional Games (C-EPCGs) and Edge Path Coalitional Games (C-EPCGs) for both directed and undirected graphs.
\end{theorem}
\begin{proof}
	
We use Theorem~\ref{th:main} to prove the statement.\\

\noindent	
C-EPCGs: For a C-EPCG $G$, it is sufficient to show that for a weight vector,  $x=(x_1,\ldots, x_{|E|})$, we can compute a minimum weight winning coalition of $G^x$. Each player (edge) $i$ has a weight $x_i$ and the minimum weight winning coalition is an $s$-$t$ simple path $P$ with the smallest weight, that is the shortest $s$-$t$ path. Use \emph{Dijkstra's Shortest Path Algorithm} to compute the shortest path $P$ from $s$ to $t$ in graph $G^{x}$ and then the minimum weight winning coalition is $E(P)$, the edges used in path $P$.\\

\noindent 
C-VPCGs: For a C-VPCG $G$, it is sufficient to show that for a weight vector,  $x=(x_1,\ldots, x_{|V|})$, we can compute a minimum weight winning coalition of $G^x$. Each player (node) $i$ has a weight $x_i$ and the minimum excess coalition is an $s$-$t$ simple path $P$ with the smallest weight, that is the shortest vertex $s$-$t$ path. 
Then compute the shortest vertex weighted path $P$ from $s$ to $t$ in graph $G^{x}$ and then the minimum weight winning coalition is $V(P)$. Dijkstra's Shortest Path Algorithm can be used to compute the shortest vertex weighted path as follows. 
	The problem can be reduced to the classic shortest path problem in the following way: 
	duplicate each vertex (apart from $s$ and $t$) with one getting all ingoing edges, and the other getting all the outgoing edges, add an internal edge between them with the node weight as the edge weight. Use the algorithm to compute the shortest vertex path $P$ from $s$ to $t$ in graph $G^{x}$.
	\end{proof}

A least core payoff of a coalitional game is not necessarily a least core payoff of the dual game. 
Therefore, we require new algorithms to compute the least core of dual coalitional path games.

\begin{theorem}\label{th:dual-epcg-costs-lc}
There exist polynomial-time algorithms to compute and verify least core payoffs of C-$EPCG^D$s and C-$VPCG^D$s  for both directed and undirected graphs.
\end{theorem}
\begin{proof}
	
	We utilize Theorem~\ref{th:main} to prove the statement.\\
	
\noindent 	
C-$EPCG^D$: For a C-$EPCG^D$ $G$, it is sufficient to show that for a weight vector,  $x=(x_1,\ldots, x_{|E|})$, we can compute a minimum weight winning coalition of $G^x$. Each player (edge) $i$ has a weight $x_i$ and the minimum weight winning coalition is an $s$-$t$ cut $P$ with the smallest weight. Use the maximum network flow algorithm~\citep[Chapter 27, ]{CLRS01a} to compute the minimum weight edge $s$-$t$ cut $C$ in graph $G^x$. This gives us the minimal winning coalition $C$ with the minimum weight.\\

\noindent 
C-$VPCG^D$: For a C-$VPCG^D$ $G$, it is sufficient to show that for a weight vector,  $x=(x_1,\ldots, x_{|V|})$, we can compute a minimum weight winning coalition of of $G^x$. Each player (node) $i$ has a weight $x_i$ and the minimum weight winning coalition is a minimum weight $s$-$t$ vertex cut. Then compute the minimum weight $s$-$t$ vertex cut in graph $G^{x}$ and then the minimum weight winning coalition is $V(P)$. It is known that 
the minimum weight vertex $s$-$t$ cut can be computed in polynomial time for directed graphs by standard network-flow methods. The network flow method to compute the minimum edge $s$-$t$ cut can be used to compute the minimum vertex $s$-$t$ cut as following. 
The problem can be reduced to the problem of min weight $s$-$t$ edge cut of an edge weight directed graph in the following way: 
duplicate each vertex (apart from $s$ and $t$) with one getting all incoming edges, and the other getting all the outgoing edges, add an internal edge between them with the node weight as the edge weight. Set the weight of all original edges as infinite (sufficiently large). We use existing algorithms to compute the minimum weight vertex $s$-$t$ cut to construct the separation oracle for the 
C-${VPCG}^D$ least core LP.\end{proof}

We note that if instead of using $s$-$t$ connectivity settings, we consider more than two terminals then some problems  such as {\sc In-$\epsilon$-Core} become NP-hard. This follows from the fact that computing a min cut for more than two terminals is NP-hard.

\eat{
A natural solution when the core is non-empty is to increase the reward to the grand coalition to ensure that all players work together and the grand coalition is stable. An external agent who increases the reward would rather increase the reward minimally to ensure that the core of the resultant game is non-empty. This concept is termed as the \emph{cost of stability} and was introduced in \citep{BMZRR09a}. From our results it follows that the cost of stability can also be computed in polynomial time.

\begin{theorem}\label{th:cos-easy}
The cost of stability of C-EPCG, C-VPCG, C-${EPCG}^D$ and C-${VPCG}^D$ can be computed in polynomial time.
\end{theorem}
\begin{proof}
	We already saw that separation oracles for the least core LPs for EPCG, VPCG, ${EPCG}^D$ and ${VPCG}^D$ can be solved in polynomial time. Since all these algorithms rely on the construction of a polynomial-time separation oracle, by Observation~1~\citep{ABH10a}, the cost of stability of all these games can be computed in polynomial time.\end{proof}

}

\section{A closer look at path coalitional games without costs}

In this section, we take a closer look at simple path coalitional games without costs. 
We will refer to the minimum size of an $s$-$t$ cut of a unweighted graph as $c_E$ if we refer to edge cuts and as $c_V$ if we refer to vertex cuts. Then we have the following theorem:

\begin{theorem}[Characterization of path coalitional games without costs]\label{th:lc-charac}
Consider an EPCG $G_{EPCG}$ and a VPCG $G_{{VPCG}}$. Then $\epsilon_1(G_{EPCG})=1-1/c_E$ and $\epsilon_1(G_{{VPCG}})=1-1/c_V$. Moreover, there are combinatorial polynomial-time algorithms to compute and verify a least core payoff of EPCGs and VPCGs.
\end{theorem}
\begin{proof}
Consider an EPNG based on graph $G$ with detection probabilities $(p_1,\ldots, p_{|E|})$. 
Let $\Delta(A)$ denote set of mixed strategies 
(probability distributions) on a finite set $A$.
The equilibrium or {\em maxmin} strategies of the interceptor
are
the solutions
$\{ x\in \Delta(E)\ 
|\ \sum_{v\in P} x_e\cdot p_e\ge val(G)\ \text{for all}\ P\in \mathcal{P}
\}$ 
to the following linear program,
which has the optimal value $val(G)$.
\begin{equation}
\label{e-LPmaxmin}
\begin{array}{ll}
\max & \alpha \\
\text{s.t.} & \sum_{e\in P} x_e\cdot p_e \geq \alpha\ 
\text{for all}\ P\in \mathcal{P}\ ,\\
& x \in \Delta(E)\ .\\
\end{array}
\end{equation}

We notice that if $p_e=1$ for $e\in E$, then LP~\eqref{e-LPmaxmin} is equivalent to the least core LP for EPCGs. It is clear that maxmin strategy $x$ of EPNG where $p_e=1$ for all $e\in E$ is equivalent to the least core payoff of EPCG corresponding to $G$. 

LP~\eqref{e-LPmaxmin} is equivalent to LP 1 in \citep{WashWood95a} if $p_e$ is set to $1$ for each edge in both LPs. This demonstrates that computing maxmin strategies of EPNG is equivalent to computing least core payoffs of EPCGs.
\citet{WashWood95a} conclude that maxmin strategy is obtained by constructing a graph $G^{c'}$ where $c'_e=1/p_e$ and then computing the minimum weight $s$-$t$ cut $S$. Each edge $e\in S$ is then given interdiction probability proportional to $c_e=1/p_e$. 
It follows that if $p_e=1$ for all $e\in E$, then $c'_e=1/p_e$, and the minimum weight $s$-$t$ cut $S$ of $G^{c'}$ is simply the min cardinality $s$-$t$ cut $S$ of $G$. 
A maxmin strategy $x$ of the interceptor, each edge in $S$ is inspected with probability $1/|S|=1/c_E$. 
Therefore for EPCG corresponding to $G$, the payoff of each simple $s$-$t$ path or equivalently minimum winning coalition has payoff $1/c_E$ and the minimum excess $-\epsilon_1$ of the EPCG is $1/c_E-1$.

We note that a similar analysis holds for VPCGs.
\end{proof}

Theorem~\ref{th:lc-charac} helps to give a correspondence between EPCG and EPNG and also between VPCG and VPNG. We note that there is no such correspondence between, for example EPNG with detection probabilities and C-EPCG. 
Theorem~\ref{th:lc-charac} helps us formulate combinatorial algorithms to compute the least core of EPCGs and VPCGs (without costs). 
The problem of computing a least core payoff reduces to computing a minimum cardinality edge cut (or vertex cut) of the graph and uniformly distributing the probability over the minimum cut. \emph{Such least core payoffs are the extreme points of the least core convex polytope and in fact any other least core payoff is a convex combination of the extreme points.}

\eat{
\begin{table}[t]
\footnotesize
\centering
\begin{tabular}{lll}
\toprule
Graph&EPCG/VPCG&EPNG/VPNG\\ \midrule

E/V&Players&pure strategy of interceptor\\
$s$-$t$ paths&MWCs&pure strategies of passer\\
graph weights&payoff&interceptor's strategy\\ 
$x$&LC payoff&maxmin strategy\\

\bottomrule
\end{tabular}
\caption{Path coalitional games \& Path intercept games}
\label{SCG-wiretap}
\end{table}
}

Our demonstrated connection of EPNGs to the corresponding coalitional EPCG in the proof of Theorem~\ref{th:lc-charac} helps examine refinements of the maxmin strategies such as the nucleolus. 


The nucleolus of a coalitional game is the unique and arguably the fairest solution concept which is guaranteed to lie in the core if the core is non-empty. The interpretation of the nucleolus in the non-cooperative setting is the maxmin strategy of the passer which not only minimizes the number of pure best responses of the interceptor but also maximizes the potential extra payoff if the interceptor does not choose the optimal strategy. 
Computing the nucleolus is a notoriously hard problem and only a handful of non-trivial coalitional games are known for which the nucleolus can be computed efficiently~\citep[see \eg][]{ALPS09a}.

We will show that for certain graph classes like series-parallel graphs, the nucleolus strategy can be computed in polynomial time (Theorem~\ref{th:epcg-nucleolus-sp}). Series-parallel graphs are an especially useful class of graphs because they are present in many settings such as electrical networks, urban grid lay-outs etc.


\begin{definition}[Series-parallel graph]
Let $G=(V,E)$ be a graph with source $s$ and sink $t$. Then $G$ is a \emph{series-parallel graph} if it may be reduced to $K_2$(a two vertex clique) by a sequence of the following operations:
\begin{enumerate}
\item replacement of a pair of parallel edges by a single edge that connects their common endpoints;
\item replacement of a pair of edges incident to a vertex of degree 2 other than $s$ or $t$ by a single edge so that $2$ degree vertices get removed.
\end{enumerate}
\end{definition}


Denote the set of edge min cuts of a graph $G$ by $\mathcal{C}(G)$. Denote by $C_e(G)$ the set $\{S\in \mathcal{C}(G) \mid e'\in S \}$.

\begin{lemma}\label{lemma:epcg}
For an undirected series-parallel graph $G=(V\cup\{s,t\},E)$, let $x$ be a least core payoff of the corresponding EPCG and let $e\in E$ be such that $C_e(G)=\emptyset$. Then $x_e=0$.
\end{lemma}
\begin{proof}
	Let $e=(a,b)\in E$ be such that $C_e(G)=\emptyset$ and assume for contradiction that there is a least core payoff of EPCG for $G$ such that $x_e>0$.
Let the graph component in series with and left of $e$ be $G_1$, the graph component in series with and right of $e$ be $G_2$, the graph component in parallel and above $e$ be $G_3$ and the graph component in parallel below $e$ be $G_4$. Since $C_e(G)=\emptyset$ there exists no edge $e'\in G_3 \cup G_4$ such that $C_e'(G)>0$. Now assume that the mincut value of $G$ is $c^*$. The mincut $C$ with size $c^*$ must either be in $G_1$ or $G_2$. We also know that since $x$ is a least core payoff, the length of the shortest $s$-$t$ path in $G$ is $1/c^*$. We show that if $x_e>0$, then a transfer of payoff from certain edge in $G_3\cup G_4 \cup \{e\}$ increases the minimum excess, thereby showing that $x$ is not a least core payoff.

If there exists no shortest $a$-$b$ path which includes $e$, then we know that $e$ is present in no coalition which gets the minimum excess. Therefore $e$ can donate its payoff uniformly to $C$ and increase the minimum excess by $x_e/c^*$. Now assume that $e$ is in one of the shortest $a$-$b$ paths. Clearly, this is not the only simple $a$-$b$ paths because if this were the case then $e$ would be a bridge be one of the mincuts. We know that mincut value of $G_3\cup G_4 \cup \{e\}$ is more than $c^*$. Let $S$ be the minimum cut of $G_3\cup G_4 \cup \{e\}$. We know that $|S|>|C|=c^*$. Then, we can show that the minimum excess of $x$ increases if $x(S)$ is distributed uniformly over $C$. Each shortest $s$-$t$ path if $G^x$ has to pass one edge in $C$ and one edge in $S$. The the weight of each edge in $C$ has increased by $x(S)/|S|$ and the length of the shortest $a$-$b$ has decreased by $x(S)/|S|$, the excess increases exactly by positive value $x(s)/|C|-x(S)/|S|$ without decreasing any other excesses.
\end{proof}

\begin{theorem}\label{th:epcg-nucleolus-sp}
The nucleolus of EPCGs for undirected series-parallel graphs can be computed in polynomial time.
\end{theorem}
\begin{proof}

We show that the problem of computing the nucleolus of EPCGs of undirected series-parallel graphs reduces to computing the parallel-series decomposition of the graph. 
There are known standard algorithms to identify and decompose series-parallel graphs~\citep[see e.g., ]{HeYe87a}. 
The reduction is based on an inductive argument in which if we know the nucleolus of two graphs $G'$ and $G''$, then we can also compute in polynomial time the nucleolus of the graph made by connecting $G'$ and $G''$ in series or parallel. 
The proof by induction is as follows:

\paragraph{Base case:} The base case is trivial. In any graph $G$ with a single edge $e$ connecting $s$ and $t$, the (pre)nucleolus $x$ gives payoff $1$ to $e$.

\paragraph{Induction:} Our induction involves two cases: attaching two graph components in series and parallel. Consider two series-parallel undirected graphs $G'$ and $G''$ and assume we already know that their nucleoli are $x'$ and $x''$ respectively. We will show that computing the nucleolus of $G$ formed by joining $G'$ and $G''$ in series and parallel is polynomial-time easy.

\begin{enumerate}
	\item
Assume we attach $G'$ and $G''$ in series to obtain $G$. Let the size of any edge mincut be $c'$ and any edge mincut be $c''$. If $c'< c''$, then by Lemma~\ref{lemma:epcg}, there is no advantage of giving payoff to any edges in $G''$. 
Therefore, the nucleolus of $G'$ is equal to the nucleolus of $G$ and we are done. 
Assume that $c'=c''$. 
We recall that the nucleolus satisfies \emph{anonymity} and \emph{covariance}~\citep{PeSu03a,Sn95a}. 
Then due to Lemma~\ref{lemma:epcg} and covariance and anonymity property of the nucleolus, we have $x=(\alpha{x'},(1-\alpha)x'')$ where $0<\alpha<1$. 
Let $m'$ and $m''$ be the smallest non-zero payoff of a player in $x'$ and $x''$ respectively. 
We then show that $x=(\alpha{x'},(1-\alpha)x'')$ is the nucleolus if $\alpha$ has the unique value for which $m'(\alpha)=m''(1-\alpha)$, i.e., $\alpha=m''/(m'+m'')$. If this were true, then   $x=(m''/(m'+m''){x'},(1-m''/(m'+m''))x'')$. In this case, the minimum excess for $x$ is $1/c'-1$ and the number of coalitions achieving this in $G$ is $|A|\times 2^{|B|}$ where $A$ is the set number of simple paths in $G^x$ and $B=\{e\in E(G)\mid C_e(G)=\emptyset\}$. We also know that the value of the second minimum excess is $1/c'-1+ (m'\cdot m'')/(m'+m'')$. 
Now assume that there exists another payoff $y=(\alpha{x'},(1-\alpha)x'')$ for some $\alpha \neq m''/(m'+m'')$ such that $y$ has a lexicographically greater excess vector than $x$. Clearly $y$ is a least core payoff of $G$. 
Then the minimum excess for $x$ is $1/c'-1$ and the number of coalition achieving this in $G$ is still $|A|\times 2^{|B|}$. 
However the second minimum excess for $y$ is less than $1/c'-1+ (m'\cdot m'')/(m'+m'')$. Therefore, $y$ has a smaller lexicographical excess vector than $x$ which is a contradiction.

 \item Consider two series-parallel undirected graphs $G'$ and $G''$ and assume we attach them in parallel to obtain $G$. Let the size of any edge mincut of $G'$ be $c'$ and the size of any edge mincut of $G''$ be $c''$. Both the mincut values can be computed in polynomial time for a graph. 
We know that the size of mincut of $G$ is $c'+c''$. 
Then due to Lemma~\ref{lemma:epcg} and covariance and anonymity properties of the nucleolus, we know that $x=(\alpha{x'},(1-\alpha)x'')$ where $0<\alpha<1$. We then show that $x=(\alpha{x'},(1-\alpha)x'')$ is the nucleolus if alpha has the unique value $c'/(c'+c'')$. Since the size of a mincut of $G$ is $c'+c''$, every least core payoff $y$ of is such that $G^y$ has shortest path $1/(c'+c'')$. We want that every shortest $s$-$t$ path which passes from $G$ to have length $1/(c'+c'')$. This is only possible if $\alpha=c'/(c'+c'')$.
\end{enumerate}
\end{proof}

We conjecture that a similar approach may help construct a polynomial-time algorithm to compute the nucleolus of VPCGs for series-parallel graphs.

\section{Conclusion}

Path coalitional games provide a simple yet rich framework to model strategic settings in the area of network security and logistics. 
In this paper we analyzed different generalizations and variants of path coalitional games and classified the computational complexity of computing different cooperative and noncooperative game solutions.\footnote{This material is based upon work supported by the Deutsche Forschungsgemeinschaft under grants BR-2312/6-1 (within the European Science Foundation's EUROCORES program LogICCC) and BR~2312/7-1. The research of Troels Bjerre S{\o}rensen was supported by EPSRC award EP/G064679/1. 
We thank Mingyu Xiao, Evangelia Pyrga, and Markus Brill for useful comments.}
One key result was a general method to compute least core payoffs of cost-based generalizations of simple games.
Many of our positive results are based on separation oracles and linear programs. It will be interesting to see if there are purely combinatorial algorithms for the same problems. Apart from the EPCGs on series-parallel graphs, the complexity of computing the nucleolus is open for all other games. For all variants of path coalitional games, we assumed that each edge/vertex is owned by a separate player. It will be interesting to see if our positive results can be extended to the more general scenario where a single player may own more than one edge or vertex.

\eat{
\section*{Acknowledgements}
This material is based upon work supported by the Deutsche Forschungsgemeinschaft under grants BR-2312/6-1 (within the European Science Foundation's EUROCORES program LogICCC) and BR~2312/7-1. The research of Troels Bjerre S{\o}rensen was supported by EPSRC award EP/G064679/1. 
We thank Mingyu Xiao and Evangelia Pyrga for useful comments.
}

\def\bibfont{\small}

\end{document}